\documentclass{article}

     \PassOptionsToPackage{numbers, compress}{natbib}
\usepackage[preprint]{neurips_2023}
\usepackage{amsmath}

\usepackage[breaklinks=true]{hyperref}
\hypersetup{colorlinks=true,%
            citebordercolor={.6 .6 .6},linkbordercolor={.6 .6 .6},%
citecolor=blue,urlcolor=black,linkcolor=blue,pagecolor=black}
\usepackage[nameinlink]{cleveref}
\Crefname{algocf}{Algorithm}{Algorithms}
\crefname{algocfline}{line}{lines}

\usepackage[utf8]{inputenc} % allow utf-8 input
\usepackage[T1]{fontenc}    % use 8-bit T1 fonts
\usepackage{url}            % simple URL typesetting
\usepackage{booktabs}       % professional-quality tables
\usepackage{amsfonts}       % blackboard math symbols
\usepackage{nicefrac}       % compact symbols for 1/2, etc.
\usepackage{microtype}      % microtypography
\usepackage{xcolor}         % colors
\usepackage{amsthm}
\usepackage{xspace}

\usepackage[ruled,vlined,linesnumbered,algonl]{algorithm2e}
\SetEndCharOfAlgoLine{}
\usepackage{breakcites}
\usepackage{scalerel}[2016/12/29] % To scale z in exponent
\usepackage{multirow}
\usepackage{tikz}

\newtheorem{definition}{Definition}
\newtheorem{theorem}{Theorem}
\newtheorem{lemma}{Lemma}
\newcommand{\veps}{\varepsilon}

\newcommand{\X}{\mathcal{X}}

\newcommand{\cI}{{\mathcal I}}

\newcommand{\kmed}{{\textsc{$k$-median}}\xspace}
\newcommand{\kmean}{{\textsc{$k$-means}}\xspace}
\newcommand{\rgath}{{\textsf{$r$-gathering}}}

\newcommand{\cO}{{\cal O}}

\newcommand{\cost}{{\mathsf{cost}}}
\newcommand{\opt}{{\mathsf{opt}}}
\newcommand{\sig}{\sigma^\star}
\newcommand{\Cstar}{C^\star}
\newcommand{\cstar}{c^\star}
\newcommand{\Gammas}{\Gamma^\star}
\newcommand{\tst}{t^\star}

\title{Universal Weak Coreset}

\author{%
  Ragesh Jaiswal and Amit Kumar
  Department of Computer Science\\
Indian Institute of Technology Delhi\\
  \texttt{\{rjaiswal, amitk\}@cse.iitd.ac.in} \\
}

\begin{document}

\maketitle

\begin{abstract}
Coresets for $k$-means and $k$-median problems yield a small summary of the data, which preserve the clustering cost with respect to any set of $k$ centers. Recently coresets have also been constructed for constrained $k$-means and $k$-median problems. However, the notion of coresets has the drawback that (i) they can only be applied in settings where the input points are allowed to have weights, and (ii) in general metric spaces, the size of the coresets can depend logarithmically on the number of points. The notion of {\em weak coresets}, which have less stringent requirements than coresets,  has been studied in the context of classical $k$-means and $k$-median problems. A weak coreset is a pair $(J,S)$ of subsets of points, where $S$ acts as a summary of the point set and $J$ as a  set of potential centers. This pair satisfies the properties  that (i) $S$ is a good summary of the data as long as the $k$ centers are chosen from $J$ only, and (ii) there is a good choice of $k$ centers in $J$ with cost close to the optimal cost.  We develop this framework, which we call {\em universal weak coresets}, for constrained clustering settings. In conjunction with recent coreset constructions for constrained settings, our designs give greater data compression, are conceptually simpler, and apply to a wide range of constrained $k$-median and $k$-means problems.

% In the context of the $k$-means/median problem, a {\em weak coreset} of a given dataset $X$, is a pair of sets $(J, S)$ such that: (i) $J$ contains $k$ centers that approximate the optimal cost within $(1+\varepsilon)$, and (ii) for any $k$ centers in $J$, the cost of the weighted set $S$, is within $(1\pm \varepsilon)$ of the cost of $X$. This is a weaker notion than a {\em coreset}, a weighted set $S$ for which (ii) holds for {\em any} $k$ centers.  The idea of a weak coreset has been used to obtain greater data compression, but its use has been limited to finding a $(1+\varepsilon)$-approximate solution for the classical $k$-means/median problems. This is because, in known constructions, the cost-preservation property holds only for center sets in $J$, which does not guarantee anything beyond containing a $(1+\varepsilon)$-approximate solution for the classical problem. In particular, such weak coresets cannot be used for constrained versions. We propose a stronger notion of a weak coreset, which we call {\em universal weak coreset}. A universal weak coreset $(J, S)$ provides the guarantee that for any target clustering of the dataset, there is a good center set in $J$ for that target clustering. In conjunction with recent coreset constructions techniques that work for constrained settings, our designs give greater data compression, are conceptually simpler, and apply to a wide range of constrained $k$-median/means problems.
\end{abstract}

\section{Introduction}
Center-based clustering problems such as \kmed and the \kmean
are important data processing tasks. Given a set of center locations $F \subset \mathcal{X}$ and a set $X \subset \mathcal{X}$ of $n$ points in a metric $(\mathcal{X}, D)$, and a parameter $k$, the goal here is to partition the set of points into $k$ {\it clusters}, say $X_1, \ldots, X_k$,  and assign the points in each cluster to a corresponding {\it cluster center}, say $c_1, \ldots, c_k \in F$ respectively,  such that the objective $\sum_{i=1}^k \sum_{x \in X_i} D(x,c_i)^z$ is minimized. Here $z$ is a parameter which is 1 for $\kmed$ and 2 for $\kmean$.
In the past decade, there has been significant effort in designing coresets for such settings. Given a \kmed or \kmean clustering instance as above, a coreset with parameter $\varepsilon$ is a weighted subset $S$ of points in the metric space with the following property: for every set $C$ of $k$ points in the metric space, the assignment cost of $X$ to $C$ is within $(1 \pm \varepsilon)$ of that of $S$. More formally, let $w(x)$ denote the weight of a point $x \in S$, and for a point $x \in X$, let $D(x, C)$ be the distance between $x$ and the closest point in $C$. Then the following condition is satisfied for every subset $C$ of $k$ points (where $z=1$ or $z=2$ depending on the clustering problem being considered): 
\begin{align}
    \label{eq:cond}
 (1- \varepsilon) \sum_{x \in S} w(x) \cdot D(x, C)^z \leq \sum_{x \in X} D(x,C)^z \leq ( 1+ \varepsilon) \sum_{x \in S} w(x) \cdot D(x,C)^z. 
 \end{align}

The notion of coresets is useful for several reasons: (i) There are efficient algorithms for constructing small-sized coresets. 
%(see table). 
Hence, some of the fastest known algorithms for \kmean and \kmed problems proceed in a step fashion: first, find a succinct coreset, and then run a less efficient algorithm on the coreset; (ii) in streaming settings, where one cannot afford to store the entire dataset, a coreset provides a summary of the data without compromising on the quality of clustering. Further, it is well known that coresets from two distinct data sets can be composed to yield a new coreset for the union of these two datasets. Hence, coresets are amenable to settings where data arrives over time; (iii) in scenarios where the set of $k$ centers may change over time, a coreset represents an efficient way of computing the clustering cost.

For most applications, the  requirements of a coreset may seem too strong. Indeed, a less stringent notion of {\em weak coreset} was defined by~\cite{fms07}. A weak coreset, with a parameter $\varepsilon$, for a point set $X$ as above is a pair $(J,S)$ of subsets of points in the metric space, with $S$ being a weighted subset of points, such  that (i) the condition~\eqref{eq:cond} is satisfied for all subsets $C$, where $|C| = k$ and $C \subseteq J$; and (ii) there is a subset $C$ of $k$ centers in $J$ such that the assignment cost of $X$ to $C$ is within $(1+\varepsilon)$ of the optimal clustering cost of $X$. The motivation for defining a weak coreset is that one could obtain weak coresets with better guarantees than a coreset. Indeed, this shall be the case in the problems considered in this work.

To understand why weak coresets may have better guarantees than coresets, we briefly discuss coreset construction techniques. Typical constructions use random sampling-based ideas. One starts with an initial set of $O(k)$ centers obtained by a fast approximation algorithm. For each of these centers $c \in C$, we partition the data into ``rings'' of geometrically increasing size around $c$. From each of these rings, one samples $poly(\frac{k}{\varepsilon})$ points and appropriately assigns them weights -- these weighted sampled points ``represent'' the points in the ring as far as $c$ is concerned, i.e., their assignment cost to $c$ is very close to that of the original set of points in the ring with high probability. 
These sampled points form the desired coreset. However, for the coreset property to hold, these sampled points must have near-optimal assignment cost for {\em every} set of $k$ centers. Since there are about $n^k$ possibilities for the choice of $k$ centers,  we need to sample $\left(poly(\frac{k}{\varepsilon})\cdot \log{n} \right)$ points from each ring to ensure the coreset property. In geometric settings, concepts  such as an $\varepsilon$-net and $\varepsilon$-centroid set have been used  to reduce the coreset size. 
However, in general  metric spaces, there are lower bounds (see \cite{bbh20}) suggesting that the size of the coreset  will have a dependency on $\log n$.

Weak coresets allow us to remove the dependency on $\log n$ even in general metric spaces. Since the near-optimal clustering guarantees need to hold with respect to $k$ centers chosen from $J$ only, the set of such possibilities reduces to $|J|^k$. Thus, a small-sized $J$ would typically imply a small-sized sample $S$ as well. Further, weak coresets allow us to maintain a near-optimal clustering in  streaming setting. 
%\ragesh{Since we are not giving any details of why $J$ can be obtained in a single pass, I am not sure whether we should mention this. What do you think? I am tentatively removing ``single-pass''}. 
Indeed, the sets $J$ and $S$ can be constructed in a streaming setting. Since the set of $k$ centers needs to be selected from $J$ only, and each can be tested with respect to $S$, we can also maintain a set of near-optimal $k$ centers in a streaming setting.

So far, our discussion has focused on the classical \kmed and \kmean settings. However, there has been significant  recent activity in the more general class of {\em constrained} clustering problems. A constrained clustering problem specifies additional conditions on a feasible partitioning of the input points into $k$ clusters. For example, the $\rgath$ problem requires that each cluster in a feasible partitioning must contain at least $r$ data points. Similarly, the well-known {\it capacitated} clustering problem specifies an upper bound on the size of each cluster.  Constrained clustering formulations can also capture various types of {\it fairness} constraints: each data point has a {\it label} assigned to it, and we may require upper or lower bounds on the number (or fraction) of points with a certain label in each cluster.
Some of these constrained problems are discussed in Section~\ref{sec:applications}.
%~\Cref{table:1} gives a list of some of these problems.

Coresets for constrained clustering settings were recently constructed by ~\cite{bfs21,focs22}. 
%\aknote{give some details}. 
Note that the standard notion of coreset is meant to preserve the cost of an assignment where points get assigned to the closest center.
This prevents using standard coresets in constrained clustering settings where a point may not necessarily get assigned to its closest center.
Recent work~\cite{bfs21,focs22} design ``assignment-preserving'' coresets that allows their use in constrained settings.
In this work, we generalize the notion of weak coresets to {\em universal weak coresets} for constrained clustering settings.
The underlying idea is the same as that of a weak coreset, i.e., we need a weighted subset $S$ of points along with a set $J$ of potential center locations. 
But now, this pair has the same guarantees as a weak coreset for {\em any} constrained clustering problem. This universal guarantee has a feature that we need not know in advance the actual constrained clustering problem being solved.

The notion of a universal weak coreset also has the following subtle application. In some specific settings, there is a distinction between known algorithms for weighted and unweighted settings. More specifically, there exist constrained clustering problems, where even if we are given a small-sized set $S$ of points, efficient algorithms for  a near-optimal set of $k$ centers with respect to $S$ are known only if the point set $S$ is unweighted.
%\aknote{give examples}. 
For example, a recent development \cite{cdk23} in the $k$-median problem in the Ulam metric has broken the $2$-approximation barrier. However, their $(2-\delta)$-approximation algorithm works only
on unweighted input permutations.
In such settings, we may not be able to efficiently find a good set of centers even if $S$ is a coreset. However, when given a weak coreset $(J, S)$, we know that we need to look for centers that are subsets of $J$ only, and we can use the cost preservation property of the weighted set $S$ to find good centers from $J$. This allows us to efficiently handle such constrained clustering problems as well.

%This allows us to prove several new results for ...\aknote{say more later, talk about weighted vs unweighted}

\paragraph{Breaking the coreset $\log{n}$ barrier} Since it is known \cite{bbh20} that the $\log{n}$ factor in the size of a coreset is unavoidable in general metric spaces, we must relax the notion of a coreset to break the $\log{n}$ barrier. 
Our notion of a universal weak coreset provides a framework for an appropriate relaxation that allows us to break the $\log{n}$ barrier.
More specifically, we relax the condition on the set $J$ to: there exists a subset $C$ of $k$ centers in $J$ such that the assignment cost of $X$ to $C$ is within $(\alpha + \veps)$ of the optimal clustering cost of $X$, where $\alpha$ is allowed to be $>1$. 
Moreover, the {\em universal} property on $J$ says that this $(\alpha+\veps)$-approximation holds with respect to {\em any} target clustering (not only the optimal Voronoi partitioning).
The property on the set $S$ remains unchanged.
We call this an $\alpha$-universal weak coreset.
Note that a $\alpha$-universal weak coreset helps to find an $\alpha$-approximate solution.
The relaxation from $(1+\veps)$ to $(\alpha+\veps)$ guarantee is not a significant compromise if $\alpha$ is the best approximation guarantee known for a constrained clustering problem, which is indeed true for several constrained problems we discuss in this paper.
On the other hand, this relaxation allows the universal weak coreset size, $(|J| + |S|)$, to be $poly(\frac{k}{\veps})$, i.e., independent of $n$.
Our main results include constructions of such universal weak coresets:

\underline{Informal result}: {\it We give a construction of a $3$-universal weak coreset for the $k$-median and a $9$-universal weak coreset for the $k$-means problem in general metric spaces (the $3, 9$ factors improve to $2, 4$ for the special case when $X \subseteq F$).
We also give a $1$-universal weak coreset construction for $k$-median/means in the Euclidean setting. All these have size $poly(\frac{k}{\veps})$.}

As applications, we discuss how to obtain an $\alpha$ approximate solution from an $\alpha$-universal weak coreset, for arbitrary versions of constrained clustering problems, such as balanced clustering, fair clustering, $l$-diversity clustering, and potentially many more.

%\ragesh{added a related work section below.}
\paragraph{Related work}
Two decades ago, coresets were introduced \cite{harpeled04} primarily as a tool to design streaming algorithms for the $k$-median/means problems. 
Subsequently, it became an independent computational object for study, and some remarkable techniques and results \cite{hk07,chen09,fms07,ls10,fl11,fss20} have been obtained.
More recent developments \cite{sw18,hjlw18,bbc19,hv20,bbv20,bjkw21,css21} have been in getting improvements on the size bounds of coresets in various metrics. 
Recent developments have also been on coresets for constrained settings \cite{bfs21,focs22}, which are most relevant to our work.

\paragraph{Organization} In the next section, we define the notion of a universal weak coreset. In Section~\ref{sec:construction}, we will see constructions of such coresets. Finally, in Section~\ref{sec:applications}, we will see applications of universal weak coresets in finding approximate solutions to several constrained clustering problems.

\section{Universal Weak Coreset}

We define the notion of universal weak coreset formally in this section. 
%
%We define the universal weak coreset that we motivated in the previous section. Let us first set up the context.
We shall use $[k]$ to denote the set $\{1, ..., k\}$. 
In the discussion, `with high probability' should be interpreted as with a probability of at least 0.99.
Let $\X$ denote a metric space with metric $D$ defined on it. We now formally define a constrained clustering problem. 
While describing an instance $\cI$, we would like to separate out the actual constraints on feasible clusterings and the underlying clustering instance. A clustering instance $\cI'$ is given by a tuple $(X, F, w, k)$, where $X$ is the set of all input points with a corresponding weight function $w: X \rightarrow \mathbb{R}^+$, 
a set $F$ of potential center locations and a value $k$, which denotes the number of clusters. 

A constrained clustering instance consists of a tuple $(X, F, w,k)$ as above and a $k$-tuple $\Gamma=(t_1, \ldots, t_k)$ of non-negative real values such that $\sum_{i \in [k]} t_i = \sum_{x \in X} w(x)$ . Intuitively, the value $t_i$ denotes the total weight of the points assigned to the $i^{th}$ cluster. However, a point in $X$ can be partially assigned to several clusters; but the sum of these partial weight assignments should equal $w(x)$. In other words, an assignment is given by a mapping $\sigma: X \times [k] \rightarrow \mathbb{R}^+$, such that $\sum_{i \in [k]} \sigma(x,i) = w(x)$ for each $x \in X$. 
An assignment $\sigma$ is said to be {\em consistent} with $\Gamma=(t_1, \ldots, t_k)$, denoted $\sigma \sim \Gamma$, if 
$\sum_{x \in X} \sigma(x,i) = t_i$ for all $i \in [k].$ Thus, the $k$-tuple $\Gamma$ denotes how the weights of the points in $X$ gets partitioned into the $k$ clusters. 
Given an instance $\cI=((X,F,w,k),\Gamma)$ of constrained clustering, and a set $C \subseteq F$ of $k$ centers, the clustering cost, denoted $\cost_z(X,w,C, \Gamma)$, where $z=1$ or $2$, defined as follows: 
\[
\cost_z(X, w, C, \Gamma) \equiv \min_{\sigma \sim \Gamma}{ \left\{\sum_{i=1}^{k} \sum_{x \in X} \sigma(x, i) \cdot D(x, c_i)^z \right\}}.
\]
Now the optimal cost of clustering over the choice of centers $C$ is denoted as follows:
$$\opt_z(X, w, \Gamma) \equiv \min_{C: C \subseteq F, |C|=k}{\{\cost_z(X, w, C, \Gamma)\}}.$$
% Given a $k$-tuple of weights $\Gamma \equiv (t_1, ..., t_k) \in (\mathbb{R}^+)^k$, an assignment $\sigma: X \times [k] \rightarrow \mathbb{R}^+$ is said to be consistent with $\Gamma$, denoted by $\sigma \sim \Gamma$, if for every $i \in [k], \sum_{x \in X} \sigma(x, i) = t_i$.
% We will use $\Gamma$ to denote how the client (weights) get partitioned into $k$ clusters with $t_i$ denoting the total client weight assigned to the $i^{th}$ center. An assignment function $\sigma$ consistent with $\Gamma$ means that cluster weights are as per $\Gamma$ if $\sigma$ is used to assign client (weights) to centers. In this paper, we will consider constrained clustering problems where the constraint is based on $\Gamma$. 
% This captures several constrained clustering problems, such as the balanced $k$-median/means problem. We will discuss these problems in Section~\ref{sec:applications}.
% For any $\Gamma \equiv (t_1, ..., t_k)$ with $\sum_i t_i = \sum_{x \in X} w(x)$, and a $k$-tuple of centers $C \equiv (c_1, ..., c_k)$, the cost of the best assignment of the client set $X$, that respects the cluster weights given by $\Gamma$, is denoted by $cost_z(X, w, C, \Gamma)$. Mathematically,
% \[
% cost_z(X, w, C, \Gamma) \equiv \min_{\sigma \sim \Gamma}{ \left\{\sum_{i=1}^{k} \sum_{x \in X} \sigma(x, i) \cdot D(x, c_i)^z \right\}}.
% \]
% Similarly, the cost of best clustering for given cluster weights $\Gamma$, is defined as 
We are now ready to define the notion of weak coresets. In the following, the parameter $z$ shall be either $1$ or 2.  We shall also fix a parameter $\varepsilon > 0$ for rest of the discussion. This should be treated as an arbitrary small but positive constant. 
\begin{definition}[$\alpha$-Universal Weak Coreset] \label{def:core}
 Given a clustering  instance $\cI=(X, F, w,k)$, an $\alpha$-universal weak coreset is a  tuple $(J, S, v)$, where $J \subset F$ is a subset of potential center locations, and $S \subset X$ is a weighted subset of points with weight function $v: S \rightarrow \mathbb{R}^{+}$ such that 
% for any constrained clustering instance, $((X,F,w,k),\Gamma)$, 
any assignment $\sigma: X\times [k] \rightarrow \mathbb{R}^+$:
 %weighted partition $(w_1, ..., w_k)$ of clients 
 %(i.e., $\forall x \in X, \sum_{i=1}^{k} \sigma(x, i) = w(x)$):
 the following condition holds with high probability: 
 %for any constrained clustering instance $((X,F,w,k),\Gamma)$ and any assignment $\sigma$ consistent with $\Gamma$ (for sake of brevity, let $w_i(x)$ denote $\sigma(x,i)$):
\begin{enumerate}
\item[(A)] $J$ contains a subset $(c_1, ..., c_k)$ with 
\[
\sum_{i=1}^{k} \sum_{x \in X} \sigma(x, i) \cdot D(x, c_i)^z \leq (\alpha + \veps) \cdot \sum_{i=1}^{k} \sum_{x \in X} \sigma(x, i) \cdot D(x, c_i^*)^z.
%\cost_z(X, w, C, \Gamma) \leq (\alpha + \veps) \cdot \opt_z(X, w, \Gamma).
\]
where $(c_1^*, ..., c_k^*)$ is the optimal center set that respects $\sigma$, i.e., $(c_1^*, ..., c_k^*) = \arg\min_{(s_1, ..., s_k)}{\left\{ \sum_{i=1}^{k} \sum_{x \in X} \sigma(x, i) \cdot D(x, s_i)^z\right\}}$
\item[(B)] For every subset $C \subseteq J$, $|C| =k$ and every $\Gamma$:
\[\cost_z(X, w, C, \Gamma) \in (1\pm \varepsilon) \cdot \cost_z(S, v, C, \Gamma).\]
\end{enumerate}
The size of a weak coreset $(J,S,v)$ is defined as  $(|J| + |S|)$.
\end{definition}

% Before we see possible constructions of a universal weak coreset, let us discuss their potential benefits.
% At the time of data $(F, X)$ availability, the constraints are unknown, meaning that $\Gamma$ is unknown since we consider constrained problems where the constraints depend on $\Gamma$.
% In other words, we do not know which constrained clustering problem needs to be solved on the dataset. 
An $\alpha$-universal weak coreset allows us to summarise the dataset so that this summary is sufficient to obtain an $(\alpha+\varepsilon)$-approximate solution to any constrained version  of the clustering problem in time that is dependent only on the size ($|J| + |S|$) of the coreset. 
This could lead to a fast approximation algorithms if the universal coreset construction is efficient and its size is independent of the data size, $n = |X| + |F|$. In the next section, we will see that this is indeed possible.
Let us see a canonical approximation algorithm that finds an $(\alpha+\veps)$-approximate solution from an $\alpha$-universal weak coreset.

\begin{theorem}\label{thm:meta}
Consider a clustering instance $(X, F, w,k)$ and let $(J, S, v)$ be an $\alpha$-universal weak coreset for it. 
Given a constrained clustering instance $((X,F,w,k), \Gamma)$, there is  algorithm $\mathcal{A}$ that, with high probability, outputs  a set of $k$ centers $C \subseteq F$ such that:
\[
\cost_z(X, w, C, \Gamma) \leq (\alpha + \varepsilon) \cdot \opt_z(X, w, \Gamma).
\]
Moreover, the running time of $\mathcal{A}$ is $\tilde{O}(|J|^k \cdot |S|)$.
\end{theorem}
\begin{proof}
The algorithm tries out all ordered subsets $C:=(c_1, \ldots, c_k)$ of size $k$ of $J$. For each such subset, one can find an assignment $\sigma: S \times [k] \rightarrow \mathbb{R}^+$ that is consistent with $\Gamma$  and minimizes $\cost_z(S,v,C, \Gamma)$. This can be done by setting up a suitable min-cost flow network. Thus, we can efficiently compute $\cost_z(X,w,C,\Gamma)$. Finally, we output the subset $C$ such that $\cost_z(X,w,C,\Gamma)$ is minimized. The desired result follows easily from the properties of a universal weak coreset. 
\end{proof}

Note that if the coreset construction is efficient ({\it i.e., polynomial in $n, k, 1/\veps$}) and the coreset size is $f(k, \veps)$,  for some function $f$, then the above theorem gives an FPT ({\it Fixed Parameter Tractable}) approximation algorithm with parameter $k$. This means that as long as $k$ is a fixed constant, the algorithm runs in polynomial time. We now give efficient constructions of universal weak coresets.

\section{Universal Weak Coreset Construction}\label{sec:construction}
We now give an algorithm for constructing coresets. Recall that 
there are two sets in the definition of a universal weak coreset: $J$ and $S$. The set $S$ represents the input points  that need to be clustered, whereas the set $J$ acts as the representative of the potential center locations.
We will construct these two sets independently using two known lines of results.

{\bf Constructing $J$:}
Let us first see the construction of the set $J$ that follows from developments on $D^z$-sampling based algorithms for the {\em list-$k$-median/means} problems~\cite{gjk20,bgjk20} -- the idea of list $\kmed$ or $\kmean$  gave a unified way of handling a large class of constrained clustering problems.

The following problem is addressed by  \cite{gjk20,bgjk20}. Given a clustering instance $(X,F,w,k)$, and a parameter $\varepsilon > 0$, output a list ${\cal L} = \{C_1, \ldots, C_\ell\}$, where each $C_i \subseteq F$ is a set of $k$ centers such that the following property is satisfied: for any partition $P_1, \ldots, P_k$ of the point set, 
%and choice of centers $c_1, \ldots, c_k$, 
there exists a set of $k$ centers  $C=(c_1', \ldots, c_k') \in {\cal L}$ such that 
$$ \sum_{i \in [k]} \sum_{x \in P_i} D(x,c_i')^z \leq (\alpha + \varepsilon) 
\sum_{i \in [k]} \sum_{x \in P_i} D(x,c_i^*)^z,$$ 
where $c_i^*$ is the optimal center for $P_i$.
The goal is to minimize the size $\ell$ of $\cal L$ (the above property needs to hold with high probability). 
To solve this problem, \cite{gjk20,bgjk20} find a suitable  set $M \subseteq F$ (using a $D^z$-sampling technique) and then iterate over all subsets of size $k$ of $M$ to generate the list $\cal L$. 
We state the relevant result from \cite{gjk20} that we shall use to construct the set $J$.\footnote{Note that the result in this particular form is not explicitly stated in \cite{gjk20} since this was not the primary goal of that work. In particular, the result stated here is a weighted version of the results in \cite{gjk20}. However, it follows  from their analysis. }

\begin{theorem}[\cite{gjk20}]\label{thm:gjk}
There is a randomised algorithm, that outputs a set $M \subset F$ of size $\left( poly(\frac{k}{\varepsilon})\right)$ with the following property: For any assignment $\sigma: X \times [k] \rightarrow \mathbb{R}^+$, with high probability, there is a set of $k$ centers $C:= \{c_1, \ldots, c_k\} \in M$ such that:
\[
\sum_{i=1}^{k} \sum_{x \in X} \sigma(x, i) \cdot D(x, c_i)^z \leq (3^z + \veps) \cdot \sum_{i=1}^{k} \sum_{x \in X} \sigma(x, i)\cdot D(x, c_i^*)^z,
\]
where $(c_1^*, ..., c_k^*)$ is the optimal set of centers that respect $\sigma$, i.e.,  $(c_1^*, ..., c_k^*) = \arg \min_{(s_1, ..., s_k)}{\sum_{i=1}^{k} \sum_{x \in X} \sigma(x, i)\cdot D(x, s_i)^z}$.
The running time of this algorithm is $O(n|M|)$. 
\end{theorem}

%The above theorem gives a very strong property. This is even stronger than what we require for our universal weak coreset. We now see how the above theorem gives us the set $J$ for a $3^z$-universal weak coreset.
It is not difficult to see that the set $M$ in the above theorem is precisely the set $J$ that we need for a $3^z$-universal weak coreset. 
%\begin{theorem}
%There is a randomised set $J \subset F$ of size $O\left( \frac{k \log{k}}{\veps^{2z+3}}\right)$ with the following property: For any partition $\Gamma \equiv (t_1, ..., t_k)$ with $\sum_i t_i = \sum_{x \in X} w(x)$, with high probability, $J$ contains $k$ centers $(c_1, ..., c_k)$ such that:
%\[
%cost_z(X, w, C, \Gamma) \leq (3^z + \veps) \cdot OPT_z(X, w, \Gamma).
%\]
%Moreover, the randomized set $J$ can be constructed in time ???.
%\end{theorem}
%The above theorem gives the construction of the set $J$ in the $\alpha$-universal weak coreset for $\alpha = 3^z$ (i.e., $3$ for $k$-median and $9$ for $k$-means).
For the special case of $X \subseteq F$ (i.e., a center can be located at any of the input points),~\cite{gjk20} gave an improved guarantee of $(2^z+\veps)$ instead of $(3^z + \veps)$. 
So, the same improvement transfers to the universal weak coreset.
For the Euclidean metric, \cite{bgjk20} used sampling ideas similar to \cite{gjk20} to give a result similar to~\Cref{thm:gjk}. However, the approximation guarantee here is $(1+\veps)$ and the size of $M$ is $(\frac{k}{\veps})^{O(\frac{1}{\veps})}$. This gives a $1$-universal weak coreset property for the set $J$ in the Euclidean setting.

%Now that we have seen constructions for set $J$, let us focus on the set $S$ in the universal weak coreset definition, a representative set for the clients. 
{\bf Constructing the set $S$:} Now we show how to construct the desired set $S$. 
Here, we build on the recent work of~\cite{focs22} in designing ``assignment-preserving coresets'' for \kmed and \kmean. Their construction works by partitioning the points 
into $\tilde{O}(k^2e^{-z})$ ``rings'' and then finding suitable (weighted) representatives from each of these rings. The latter procedure requires clever random sampling techniques. The selected representatives $S_R$ from a particular ring $R$ satisfy the following condition: for any set of $k$ centers $C$, the assignment costs to $C$ of all the points in the ring $R$ is close to that of $S_R$. But one would like this property to hold for all $n^k$ possible ways of choosing $C$. Thus, one needs to apply union bound over all such possibilities, which results in a multiplicative factor of $k\log{n}$ in the representative size from each ring.\footnote{Coreset construction previous to \cite{focs22} (e.g., \cite{bfs21}) had a $\log{n}$ factor also coming from the number of rings. This bottleneck was removed in \cite{focs22}.}
%There is evidence \cite{bbh20} to suggest that a factor of $\log{n}$ is unavoidable as far as classical coreset size in general metric spaces is concerned.
%To break this $\log{n}$ barrier we need to relax the notion of a coreset, and our universal weak coreset provides a framework for an appropriate relaxation. 
As mentioned earlier, one hopes to avoid this barrier by constructing weak coresets. 
%Note that, within the $\alpha$-universal weak coreset framework, it is sufficient to preserve the cost over all possible $k$ center sets from the set $J$.
Here, number of possible choices for the set $C$ reduces to $|J|^k$ instead of $n^k$. So, the $k\log{n}$ factor needed in the size of the sampled set $S_R$ from each ring $R$ gets replaced by $k \log{|J|}$. The trade-off is that instead of the classical coreset allowing a $(1+\veps)$-approximate solution, the $\alpha$-universal weak coreset only allows a $(\alpha + \veps)$-approximate solution. This is not a significant compromise
if $\alpha$ is the best approximation guarantee known for a constrained clustering problem, which happens to be true for several cases.
We now formally state the result from \cite{focs22} that we shall use to construct the set $S$ for our universal weak coreset.

\begin{theorem}[\cite{focs22}]
Consider a clustering instance $(X,F,w,k)$ and a parameter $\delta \in (0,1)$. 
There is a randomised algorithm to construct a weighted set $T \subset X$ of size $O\left(poly(\frac{k}{\veps}) \cdot \log{\frac{1}{\delta}} \right)$ with weight function $v: T \rightarrow \mathbb{R}^+$ that satisfies the following property: 
%given a clustering instance $((X,F,w,k), \Gamma)$ and 
given a set $C$ of $k$ centers, 
\[
\forall \Gamma, \cost_z(X, w, C, \Gamma) \in (1\pm \veps) \cdot \cost_z(T, v, C, \Gamma),
\]
holds with probability at least $(1-\delta)$. 
Moreover, the running time of the algorithm is $O(n|T|)$.
\end{theorem}

The construction of the desired set $S$ using the above result follows from a direct application of union bound over the choice of $k$ center sets in the set $J$.

\begin{theorem}
\label{thm:consS}
Consider a clustering instance $(X,F,w,k)$. 
There is a randomized algorithm for constructing a weighted set $S \subset X$ of size $\left(poly(\frac{k}{\veps}) \cdot \log{|J|}\right)$ with weight function $v: S \rightarrow \mathbb{R}^+$ such that  
%for any clustering instance $(X,F,w,k)$,  
the following event happens with high probability: for every choice of $C$ centers from $J$, $|C|=k$, and every $\Gamma$:
\[
%\textrm{For every $(c_1, ..., c_k)$ with $\forall i, c_i \in J$}, 
\cost_z(X, w, C, \Gamma) \in (1\pm \veps) \cdot \cost_z(S, v, C, \Gamma).
\]
Moreover, the running time of the algorithm is $O(n|S|)$. 
\end{theorem}
% \begin{proof}
% TO BE WRITTEN.
% \end{proof}

%Now that we have seen the construction of both sets $J$and $S$ required for the universal weak coreset, we can now summarise our results.

The following results now follow from~\Cref{thm:consS} and the discussion after~\Cref{thm:gjk}:

\begin{theorem}[Main theorem: Metric \kmed]\label{thm:kmed3}
There is a $3$-universal weak coreset $(J, S,v)$ of size $\left( poly(\frac{k}{\veps})\right)$ for \kmed (i.e., $z=1$) objective in general metric spaces. The time to construct such a coreset is $O(n \cdot (|J| + |S|))$. 
\end{theorem}

For the special case $X \subseteq F$, the guarantee in the above theorem improves from $3$ to $2$.

\begin{theorem}[Main theorem: Metric \kmean]\label{thm:kmean9}
There is a $9$-universal weak coreset $(J, S,v)$ of size $\left( poly(\frac{k}{\veps}) \right)$ for \kmean (i.e., $z=2$) objective in general metric spaces. The time to construct such a coreset is $O(n \cdot (|J| + |S|))$.
\end{theorem}

For the special case $X \subseteq F$, the guarantee in the above theorem improves from $9$ to $4$.

\begin{theorem}[Main theorem: Euclidean \kmed/\kmean]\label{thm:kmm1}
There is a $1$-universal weak coreset of size $\left( poly(\frac{k}{\veps})\right)$ for \kmed and \kmean objectives in the Euclidean metric. The time to construct such a coreset is $O \left(n (k/\veps)^{O(\frac{1}{\veps})}\right)$.
\end{theorem}

In the following section, we see applications of the results above.

\section{Applications}\label{sec:applications}
% In general, a universal weak coreset $(J, S)$ for $(F, X, k)$ may be interpreted as a compressed version where $F$ is compressed to $J$ and $X$ is compressed to $S$. So, $(J, S, k)$ may be considered a compressed version of $(F, X, k)$. 
In this section, apply the  universal weak coreset constructions to solve constrained versions of the \kmed and \kmean problems. 
As mentioned earlier, we can view a universal weak coreset as a compression of the original dataset. 
There are two ways of applying universal  weak coresets~: (i) Execute a known algorithm for the specific constrained problem on the compressed instance $(J, S, v, k)$ instead of $(F, X,w, k)$, and (ii) Use the meta-algorithm defined in Theorem~\ref{thm:meta} with appropriate modifications. We now discuss some specific examples.

\subsection{Clustering with size-based constraints}
We consider constrained clustering problems where, besides optimizing the objective function, 
there are  constraints on the size of the clusters. For example, the $r$-gathering problem requires  a lower bound of $r$ on the size of every cluster. Similarly, the capacitated clustering problem has an upper bound on cluster size. These constraints try to capture a ``balance'' property that limits the variance in the cluster sizes. 
We can model such size-constrained problems using the balanced \kmed or \kmean problem. 
Here,  in addition to $(F, X,w, k)$, an instance also specifies tuples $(l_1, ..., l_k)$ and  $(u_1, ..., u_k)$; where $l_i$ and $u_i$ are the lower and upper bound on the total weight of the $i^{th}$ cluster, respectively.
%So, these $l_i$'s and $u_i$'s define the version of the constrained problem. 
For example, the $r$-gathering problem is obtained by setting $l_i=r, u_i = \infty$ for all $i \in [k]$. 
Let us see how the $3$-universal weak coreset for \kmed objective from Theorem~\ref{thm:kmed3} helps obtain a $3$-approximation algorithm  for any instance of the balanced \kmed problem (the extension to balanced \kmean is analogous).

\begin{theorem}
\label{thm:applicationbalanced}
Let $(J, S,v)$ be a $3$-universal weak coreset for an input instance $(F, X, w, 
  k)$.
  %\footnote{For simplicity we do not mention the weight functions $w$ and $v$ explicitly}
Let $\cI = (F,X,w,k,(l_1, ..., l_k), (u_1, ..., u_k))$ be an instance of the balanced \kmed problem.  Then there is a randomized algorithm $\mathcal{A}$, that with high probability, outputs a $k$ center set $C$ that is a $(3+\veps)$-approximate solution for $\cI$. The running time of $\mathcal{A}$ is $\tilde{O}(|J|^k \cdot |S|)$.\footnote{The overall running time of the approximation algorithm, including the time to construct the universal weak coreset is $n \cdot poly(\frac{k}{\veps}) + \left(\frac{k}{\veps}\right)^{O(k)}$.}
%\aknote{say what the final running time is.}
\end{theorem}
\begin{proof}
Consider an optimal solution to $\cI$, and let $\sigma(x,i)$ denote the weight of point $x$ assigned to cluster $i$. 
Define  $\Gamma := (\sum_x \sigma(x,1), ..., \sum_x \sigma(x,k))$. 
From the $3$-universal weak coreset property, there is a subset $C$ of $J$, $|C|=k$, such that 
$\cost_1(X,w,C,\Gamma) \leq (3+\varepsilon) \opt(X,w,\Gamma).$
% we know that there are $k$ centers $(c_1, ..., c_k)$ in $J$ such that $\sum_i \sum_x w_i(x) \cdot D(x, c_i)\leq (3+\veps) \cdot \sum_i \sum_x w_i(x) \cdot D(x, c_i^*)$, where $(c_1^*, ..., c_k^*)$ is the optimal $k$ center set for weighted partition $(w_1, ..., w_k)$. 
Moreover,  the set $S$ has the property that $\cost_1(X, w, C, \Gamma)$ and $\cost_1(S,v,C, \Gamma)$ are within $(1+\varepsilon)$ factor of each other. 
This implies  that if we try all possible choice of $k$ centers $C$ from $J$ and, for each such $C$, find $\opt_1(S,v,C, \Gamma)$, then we can compute $\opt_1(X,w, \Gamma)$ within $(3+\varepsilon)$ approximation factor. 

The remaining issue is how to compute $\opt_1(S,v,C, \Gamma)$ for a given choice of $C.$ We do not know $\Gamma$ here, but we can find the tuple $\Gamma'$ for which $\opt_1(S,v,C, \Gamma')$ is minimized. Indeed, we can set up a minimum cost flow network where we would like to assign the points fractionally to the centers in $C$, and for each center in $C$, we can assign lower and upper bound (i.e., $l_i$ and $u_i$) for the amount of weight assigned to it. Solving this min-cost flow problem shall yield the optimal choice of $\Gamma'$. 
Minimizing over $C \subseteq J, |C|=k$, we can find $\opt_1(X,w,\Gamma)$.
%Thus, for each $C$, we can find $\opt_1(X,w,\Gamma)$ over all tuples $\Gamma$ satisfying the balance condition. 

Combining the above ideas yield a $(3+\varepsilon)$ approximation algorithm. 
\end{proof}

The $(9+\veps)$-approximation for arbitrary balanced versions of the \kmean problem in general metrics follows on similar lines using the $9$-universal weak coreset from Theorem~\ref{thm:kmean9}.
Similarly, a $(1+\veps)$-approximation for arbitrary balanced versions of the \kmed and \kmean problems in Euclidean metrics can be obtained using $1$-universal weak coreset from Theorem~\ref{thm:kmm1}.

\subsection{Fair clustering and other labeled versions}

We now consider constrained clustering problems where points have labels, i.e., 
suppose we are given a label set $L := \{1, ..., m\}$, and each point $x$ has a label $\ell(x) \in L$ associated with it. 
Labels can capture disparate scenarios where every client may be part of multiple (overlapping) groups ({\it e.g., groups based on gender, ethnicity, age, etc.}). Every unique combination of groups gets assigned a different label, so $m$ denotes the number of distinct combinations of groups to which a point can belong.
For a label $j \in L$, let $X_j$ denote the set of points that are assigned label $j$. Consider a clustering instance $(X,F,w,k, \ell),$ where we have also incorporated the label mapping. The corresponding fair clustering instance $\cI$ is specified by an additional list of 
of $k$ pairs, namely,  $(\alpha_1, \beta_1), ..., (\alpha_k, \beta_k)$. 
An optimal solution needs to find a set of $k$ centers, and an  assignment $\sigma: X \times [k] \rightarrow \mathbb{R}^+$ for all $x \in X$,  such that: 
\begin{enumerate}
\item[(i)] For every $j \in [m]$ and $i \in [k]$, $\frac{\sum_{x \in X_i} \sigma(x, i)}{\sum_{x \in X} \sigma(x, i)} \in [\alpha_i, \beta_i]$, i.e., for every group, the fraction of weights assigned to the $i^{th}$ cluster is in the range $[\alpha_i, \beta_i]$. This captures various fairness notions for points that  may belong to a particular group.  
\item[(ii)] The assignment cost, i.e., 
$\sum_{i=1}^{k} \sum_{x \in X} \sigma(x, i) \cdot D^z(x, c_i)$, is minimized. 
%%The following cost is minimized: $\fcost_z(X, w, C) \equiv \sum_{i=1}^{k} \sum_{x \in X} \sigma(x, i) \cdot D^z(x, c_i)$.
\end{enumerate}

%The above fair clustering problem belongs to a larger class of labeled clustering problems where the client can be partitioned into labeled sets $X_1, ..., X_m$. 

Our definition of universal weak coreset is for the case $m=1$, i.e., points have only one label, which may be interpreted as the unlabeled case.
However, we need to extend the notion of universal weak sets  to multi-label settings. 

Towards this, we recall that the set $J$ constructed in the previous section (Theorems~\ref{thm:kmed3} and \ref{thm:kmean9}) satisfies the following property: 
%consider an optimal solution to the above instance $\cI$. Let $l_i$ denote the total weight assigned to the $i^{th}$ cluster, and define $\Gamma := (l_1, \ldots, l_k)$. Then we know that there is a set $C$ of $k$ centers in $J$ for which $\cost_z(X,w,C,\Gamma)$ is within $(\alpha+\varepsilon)$-factor of the optimal assignment cost $\opt(X,w,\Gamma)$. 
for any assignment $\sigma: X \times [k] \rightarrow \mathbb{R}^+$, there is a $(3^z + \veps)$-approximate center set in $J$ with respect to $\sigma$.
More specifically, let $\sigma^*$ denote the optimal assignment and let $C^* \equiv (c_1^*, ..., c_k^*)$ denote the optimal $k$ centers that respects $\sigma^*$. 
%Consider the optimal weighted partition $(w_1^*, ..., w_k^*)$ given by $\forall i, x, w_i^*(x) = \sigma^*(x, i)$. 
The property on set $J$ says that there exists $k$ centers $C \equiv (c_1, ...., c_k)$ such that $\sum_i \sum_x \sigma^*(x, i) \cdot D(x, c_i)^z \leq (3^z + \veps) \cdot \sum_i \sum_x \sigma^*(x, i) \cdot D(x, c_i^*)^z$. 
This means that as long as our set $S$ has the property that for any assignment respecting the group constraint, the corresponding assignment cost to any $C \subseteq J, |C|=k$ is about the same as that of the point set $X$, we will have a $3^z$-universal weak coreset for the fair clustering problem as well. 
Here, we note that  we can execute the coreset construction from~\Cref{thm:consS} separately on each group and take a union of the corresponding coresets obtained. This larger set acts a coreset for the labeled dataset. We now formalize these ideas. 
First, we extend the notion of a universal weak coreset to the multi-labeled setting. In the unlabeled version, as assignment of weights to a centers in a set $C$ was characterized by a tuple $\Gamma$ of size $k$. Since we have $m$ labels now, such an assignment needs to be specified for each label. In other words, we now consider tuples $\Gamma$ of length $mk$, i.e., $\Gamma := (t_{1, 1}, ..., t_{1, m}, t_{2, 1}, ..., t_{2, m}, ..., t_{k, 1}, ..., t_{k, m})$, where $t_{i,j}$ is meant to denote the total weight of points with label $j$ assigned to the $i^{th}$ cluster. We can define an assignment $\sigma$ analogously as a map $X \times [k]  \rightarrow \mathbb{R}^+$. We say  that $\sigma$ is consistent with $\Gamma$, i.e., $\sigma \sim \Gamma$ if for every label $j$ and cluster $i$, $\sum_{x \in X_j} \sigma(x,i) = t_{i,j}$. Similarly, for a set of centers $C$, define $\cost_z(X,w,C,\Gamma)$ as 
\[
\cost_z(X, w, C, \Gamma) \equiv \min_{\sigma \sim \Gamma}{ \left\{\sum_{i=1}^{k} \sum_{x \in X} \sigma(x, i) \cdot D(x, c_i)^z \right\}}.
\]
Again, $\opt_z(X,w, \Gamma)$ can be defined as the optimum cost over all choices of centers $C$. Now the definition of a universal coreset $(J,S,v)$ in this setting is analogous to that in~\Cref{def:core} -- we need to satisfy conditions~(A) and~(B).

% A general constrained problem in the multi-labeled setting has constraints on the weights from every labeled set assigned to a cluster. 

% So, if we denote the cluster weight from the set with label $j$ assigned to the $i^{th}$ center as $t_{i, j}$, then the constraints can be stated in terms of the $km$ tuple $\Gamma \equiv (t_{1, 1}, ..., t_{1, m}, t_{2, 1}, ..., t_{2, m}, ..., t_{k, 1}, ..., t_{k, m})$ ({\it instead of a $k$-tuple as in the unlabeled case}). 
% We can reuse the notation and definition of the unlabeled universal weak coreset to define the multi-labeled one by simply replacing the $k$-tuple $\Gamma$ with the above $km$ tuple, (i.e., the assignment consistency $\sigma \sim \Gamma$ is checked with a $km$ tuple instead of $k$ tuple). For simplicity, we do not write the definition explicitly here.

\begin{theorem}
There is a $(3^z + \veps)$-universal weak coreset of size $\left(m \cdot poly(\frac{k}{\veps}) \right)$ for constrained clustering in the multi-labeled setting.
\end{theorem}
\begin{proof}
The set $J$ is constructed as in~\Cref{sec:construction}. In order to construct the set $S$, we apply~\Cref{thm:consS} to each of the sets $X_1, \ldots, X_m$ independently to obtain sets $S_1, \ldots, S_m$. Finally, $S := S_1 \cup \ldots \cup S_m$. The desired result follows from the properties of universal coreset. 
\end{proof}

Let us now see why a $(3^z + \veps)$-universal weak coreset can be used for obtaining a $(3^z + \veps)$-approximate solution for multi-labeled constrained clustering problem 
in FPT time (fixed-parameter tractable time). We state the result for the \kmed objective in general metric spaces. Similar results will hold for \kmean in general metric spaces (i.e., $(9+\veps)$-approximation) and $\kmed$ or $\kmean$ objectives in Euclidean spaces (i.e., $(1+\veps)$-approximation).

\begin{theorem}
\label{thm:labelcons}
Let $(J, S,v)$ be a $3$-universal weak coreset for a multi-labeled clustering  instance $(F, X, w, k, \ell)$. Consider an instance $\cI$ of the co  nstrained clustering problem specified by set of pairs 
$\{(\alpha_1, \beta_1), ..., (\alpha_k, \beta_k)\}$. 
Then there is a randomize algorithm $\mathcal{A}$, which on input $\cI$ and $(J,S,v)$ outputs a  $(3+\veps)$-approximate solution with high probability. The running time of $\mathcal{A}$ is $|J|^k \cdot (mk)^{O(mk)} \cdot n^{O(1)}$. 
\end{theorem}
\begin{proof}
The proof follows the same line as for the unlabeled case. We try all $|J|^k$ possible $k$ centers $(c_1, ..., c_k)$ from $J$ and solve the ``assignment'' problem: find the best fair assignment for the given $(c_1, ..., c_k)$. Our $3$-universal weak coreset guarantees the existence of a $(3+\veps)$-approximate solution within $J$. So, if we can solve the assignment problem optimally, we can find that $(3+\veps)$-approximate solution in $J$. Such an assignment algorithm was given by~\cite{bfs21} (see Theorem~8.2). The running time of this assignment finding algorithm is $(mk)^{O(mk)} \cdot n^{O(1)}$.
\end{proof}

\paragraph{$l$-diversity clustering} Another well known constrained clustering problem in the labelled setting is the  $l$-diversity problem.  Here the goal is to cluster the point set $X$ into clusters $(X_1, ..., X_k)$  such that each cluster has at least $1/l$ fraction of the points from each of the labels. Again, the goal is to minimize the \kmed or \kmean assignment cost. 

As above, we can use the universal weak coreset construction from~\Cref{thm:labelcons} to obtain a $(3^z + \veps)$-approximation algorithm for this problem. Here, we can use  algorithm of~\cite{dx20} to solve the corresponding assignment problem.

\subsection{Discussion and Open Problems}
Classical coresets come with the promise that they will help obtain an approximate solution to the \kmean or \kmed objective in a metric space.  
This promise holds for most known metric spaces. 
However, there are certain  metrics where a specific approximation guarantee cannot be obtained using a classical coreset. 
The reason is that the approximation algorithm that gives that specific approximation guarantee does not work on weighted inputs. Note that a classical coreset is a weighted set.
For example, a recent development~\cite{cdk23} in the \kmed problem in the Ulam metric has broken the $2$-approximation barrier. 
However, their $(2-\delta)$-approximation algorithm works only on unweighted input permutations.
So, the classical coreset framework does not help in this setting. On the other hand, the universal weak coreset framework may still be applicable. The reason is that even though we cannot run the approximation algorithm on the set $S$ to find a good center set, we can use $S$ to locate a good center set from $J$ using the cost preservation property of $S$.
So, an interesting open question is whether there is a $(2-\delta)$-universal weak coreset for the Ulam $k$-median problem. 
In general, in cases where the guarantee of the set $S$ is limited to cost preservation, i.e., $S$ represents the data only in a limited sense, a universal weak coreset is a more appropriate object to use.
It will be interesting to see if there are problems other than the Ulam $k$-median problem with this property.

Note that there are one pass streaming algorithms for constructing the set $S$ because coresets have composability property~\cite{chen09}; and there is a constant-pass streaming algorithm for constructing the set $J$ (the algorithm for constructing $M$ in~\Cref{thm:gjk} can be implemented in streaming settings). Thus both $J$ and $S$ can be constructed in a constant pass streaming setting. We leave it an open problem to design a single-pass streaming algorithm for a universal weak coreset. 

Although we give 3-universal weak coreset constructions of size independent of any function of $n$ for \kmed (and a similar result for \kmean), it remains an open problem to construct an $\alpha$-universal weak coreset, $\alpha < 3$,  with such a guarantee, even for general metric spaces. This will help obtain a better than $3$ approximation algorithm for several constrained $k$-median problems for which the best-known approximation bound is $3$ (similarly a better than $9$-approximation for \kmean).

\bibliographystyle{plain}
{\small \bibliography{references}}

\appendix

\section{Detailed Proofs}
In the main write-up, we outlined proofs of all the results. This section illustrates detailed proof of one of the results,~Theorem~8. The detailed proofs of other results stated in Section~\ref{sec:applications} can be carried out analogously.

\subsection{Proof of Theorem~8}
We restate the theorem below.
\setcounter{theorem}{7}
\begin{theorem}
Let $(J, S,v)$ be a $3$-universal weak coreset for an input instance $(F, X, w, 
  k)$.
Let $\cI = (F,X,w,k,(l_1, ..., l_k), (u_1, ..., u_k))$ be an instance of the balanced \kmed problem.  Then there is a randomized algorithm $\mathcal{A}$, that with high probability, outputs a $k$ center set $C$ that is a $(3+\veps)$-approximate solution for $\cI$. The running time of $\mathcal{A}$ is $\tilde{O}(|J|^k \cdot |S|)$.\footnote{The overall running time of the approximation algorithm, including the time to construct the universal weak coreset is $n \cdot poly(\frac{k}{\veps}) + \left(\frac{k}{\veps}\right)^{O(k)}$.}
\end{theorem}
\begin{proof}
Consider an optimal solution $\cO$ to the instance $\cI$ -- note that this solution respects the bounds $(l_1, \ldots, l_k)$ and $(u_1, \ldots, u_k)$, i.e., the total weight assigned to the $i^{th}$ cluster lies in the range $[l_i, u_i]$ for all $i \in [k]$. 
Let $\sig$ denote the assignment corresponding to $\cO$, i.e., for a point $x$ and cluster $i$, $\sig(x,i)$ denotes the weight of this point assigned to cluster $i$. Let $\Cstar := (\cstar_1, \ldots, \cstar_k)$ denote the set of $k$ centers in this solution, where $\cstar_i$ denotes the center of cluster $i$ in $\cO$. Let $\opt(\cI)$ denote the cost of the solution $\cO$. 
Let $\Gammas := (\tst_1, ..., \tst_k)$ denote
the actual weight assigned to each of the clusters by $\cO$, i.e., 
$\tst_i = \sum_{x \in X} \sig(x,i)$ for all $i \in [k].$
These quantities satisfy the following properties:
\begin{enumerate}
\item[(I)] $\forall i \in [k], \tst_i \in [l_i, u_i]$.
\item[(II)] $\sig = \arg\min_{\sigma \sim \Gammas}{\cost_1(X, w, \Cstar, \Gammas)}$.
\item[(III)] $\opt(\cI) = \sum_i \sum_x \sig(x, i) \cdot D(x, \cstar_i)= \cost_1(X, w, \Cstar, \Gammas)$
\end{enumerate}
From the $3$-universal weak coreset property, we know that $(J, S, v)$ satisfies the following properties (with high probability): 
\begin{enumerate}
\item[(A)] There are $k$ centers $(c_1, ..., c_k)$ in $J$ such that $\sum_i \sum_{x \in X} \sig(x, i) \cdot D(x, c_i) \leq (3+\veps) \cdot \opt(\cI)$.
\item[(B)] For every choice of $k$ centers $C \equiv (c_1, ..., c_k)$ in $J$ and every $\Gamma$, $\cost_1(X, w, C, \Gamma) \in (1 \pm \veps) \cdot \cost_1(S, v, C, \Gamma)$.
\end{enumerate}
The algorithm for finding an approximate solution for the constrained instance $\cI$ iterates over all subsets of $k$ centers from $J$, and for every such set $(c_1, ..., c_k)$, finds the optimal feasible assignment. 
For a fixed set of $k$ centers $(c_1, ..., c_k)$, an optimal feasible assignment can be found by setting up a min-cost flow network with appropriate lower and upper bounds.
Finally, it outputs the center set with the least cost. The pseudocode for this algorithm is given below.
 \begin{algorithm}[H]
  \caption{Algorithm for balanced clustering using coreset $(J, S, v)$}
  \label{algo:cluster}
    {\bf Input:}  $\cI := (J, S, v, k, (l_1, ..., l_k), (u_1, ..., u_k))$\;
    $mincost = \infty$; $mincenters = \{\}$\;
     \For {every set of $k$ centers $(c_1, ..., c_k)$ from $J$}{
         Set up the appropriate min-cost flow network with lower and upper bounds\;
         Let $c$ denote the optimal cost of the flow network\;
         \If {$mincost < c$}{$mincost = c$; $mincenters = (c_1, ..., c_k)$}
     } 
     {\bf return}($mincenters$)
\end{algorithm}
 Let us see why the center set returned by the above algorithm is a $(3+\veps)$-approximate solution for the balanced clustering instance $\cI$. Let $\Lambda$ denote the set of feasible assignments for this instance, i.e., for any weighted set $B$ (where the weight of a point $x$ is given by $w(x)$), 
 $$ \Lambda_B := \left\{ \sigma: B \times [k] \rightarrow \mathbb{R}^+ \Big|  \sum_{i \in [k]} \sigma(x,i) = w(x) \, \, \forall x \in B, \quad \text{and} \, \sum_{x \in B} \sigma(x,i) \in [l_i, u_i] \, \, \forall i \in [k]\right\}.$$

Given a set of centers, we can define the cost of the optimal feasible assignment as follows. For any center set $C' = (c_1', ..., c_k')$, and weighted set $B$ as above, define
\begin{eqnarray*}
\Psi(B, w, C') &:=& \min_{\sigma \in \Lambda_B} {\psi(B, w, C', \sigma)}, \textrm{ where } \psi(B, w, C', \sigma) = {\left\{ \sum_{i \in [k]
} \sum_{x \in B} \sigma(x, i) \cdot D(x, c_i')\right\}}  
%\\
%\Psi(X, w, C') &=& \min_{\sigma:X\times[k]\rightarrow \mathbb{R}^+} {\psi(X, w, C', \sigma)}, \textrm{ where } \psi(X, w, C', \sigma) = {\left\{ \sum_i \sum_{x \in X} \sigma(x, i) \cdot D(x, c_i')\right\}}
\end{eqnarray*}
In our applications, $B$ will either be the set $X$ with weights $w$, or the set $S$ with weights $v$. 
From the property on set $J$, we know that there is a set of $k$ centers $\bar{C} = (\bar{c}_1, ... , \bar{c}_k)$ in $J$ such that:
\begin{equation}\label{eqn:1}
\psi(X, w, \bar{C}, \sig) = \sum_{i \in [k]} \sum_{x \in X} \sig(x, i) \cdot D(x, \bar{c}_i) \leq (3+\veps) \cdot \opt(\cI).
\end{equation}
We will use the following two lemmas to show the approximation guarantee on the solution returned by the algorithm. 

\begin{lemma}\label{lem:1}
For any set of $k$ centers $C$, any assignment $\sigma \in \Lambda_X$ with $\Gamma := (\sum_x \sigma(x, 1), ..., \sum_x \sigma(x, k))$, $\cost_1(X, w, C, \Gamma) \leq \psi(X, w, C, \sigma)$.
\end{lemma}
\begin{proof}
The proof follows from the definition of functions $\psi$ and $\cost_1$.
\end{proof}

\begin{lemma}
For any center set  $C' = (c_1', ..., c_k')$, $C' \subseteq J$,  if $\Psi(S, v, C') \leq \Psi(X, w, C)$, then $\Psi(X, w, C') \leq (3+15\veps) \cdot \opt(\cI)$.
\end{lemma}
\begin{proof}
We will prove this by contradiction. Assume for the sake of contradiction that there is a $C'$ with $\Psi(S, v, C') \leq \Psi(X, w, \bar{C})$ but $\Psi(X, w, C') > (3+15\veps) \cdot \opt(\cI)$. We have:
\begin{eqnarray}\label{eqn:2}
\Psi(S, v, \bar{C}) &\leq& \cost_1(S, v, \bar{C}, \Gammas) \nonumber\\
&\leq& (1 + \veps) \cdot \cost_1(X, w, \bar{C}, \Gammas)  \quad \textrm{(using property (B) on $\bar{C}$)} \nonumber\\
&\leq& (1 + \veps) \cdot \psi(X, w, \bar{C}, \sig)  \quad \textrm{(Using Lemma~\ref{lem:1})} \nonumber\\
&\leq& (3 + \veps) \cdot (1 + \veps) \cdot \opt(\cI) \quad \textrm{(Using eqn. (\ref{eqn:1}))} \nonumber\\
&\leq& (3+5\veps)\cdot \opt(\cI).
\end{eqnarray}
Let $\sigma' = \arg\min_{\sigma}\psi(S, v, C', \sigma)$ and $\Gamma' = (\sum_x \sigma'(x, 1), ..., \sum_x \sigma'(x, k))$.
Using $\Psi(S, v, \bar{C}) \geq \Psi(S, v, C')$, we get the following:
\begin{eqnarray*}
\Psi(S, v, \bar{C}) &\geq& \Psi(S, v, C') \\
&=& \psi(S, v, C', \sigma') \\
&=& \cost_1(S, v, C', \Gamma') \\
&\geq& \frac{1}{1 + \veps} \cdot \cost_1(X, w, C', \Gamma')  \quad \textrm{(using property (B) on $C'$)}\\
&\geq& \frac{1}{1 + \veps} \cdot \Psi(X, w, C') \\
&\geq& \frac{3 + 15 \veps}{1+\veps} \cdot \opt(\cI) > (3+5\veps)\cdot \opt(\cI).
\end{eqnarray*}
This contradicts eqn. (\ref{eqn:2}).
\end{proof}
The theorem follows from the above lemma since we output the center set $C'$ with the least value of $\Psi(S, v, C')$.\footnote{The $(3+15\veps)$-approximation, instead of $(3+\veps)$-approximation, can be handled by using $\veps' = \veps/15$ instead of $\veps$.} 
\end{proof}

The Proof of Theorem 10 in the multi-labeled setting follows on similar lines as the above proof of Theorem 8.

\end{document}